\documentclass[12pt]{iopart}
\usepackage{graphicx}
\usepackage{amsthm}
\usepackage{cite}
\usepackage{lipsum}
\usepackage{verbatim} 
\usepackage{hyperref}
\usepackage{color}
\usepackage{bbold}
\usepackage{epstopdf}

\newtheorem{proposition}{Proposition}

\newtheorem{theorem}{Theorem}

\newcommand{\ket}[1]{\left\vert#1\right\rangle}
\newcommand{\bra}[1]{\left\langle#1\right\vert}

\def\bra#1{\langle #1|}
\def\ket#1{\left|#1 \right>}

\def\Tr{\mbox{Tr}}

\begin{document}
\title{Disturbance-Based Measure of Macroscopic Coherence}
\author{Hyukjoon Kwon, Chae-Yeun Park, Kok Chuan Tan, and Hyunseok Jeong$^*$}
\address{Center for Macroscopic Quantum Control, Department of Physics and Astronomy, Seoul National University, Seoul, 151-742, Korea}
\ead{h.jeong37@gmail.com}
\date{\today}

\begin{abstract}
We propose a measure of macroscopic coherence based on the degree of disturbance caused by a coarse-grained measurement.
Based on our measure, we point out that recently proposed criteria of macroscopic coherence may lead to inconsistent results when considering certain states such as a product of  microscopic superpositions.
An inequality relation is proved that relates the Wigner-Yanase-Dyson skew information and the measurement disturbance, providing arguments as to why our approach is able to rule out such inconsistencies.
We show that our measure can also quantify the fragility of a quantum state to a certain type of decoherence.
Our work provides a general framework of quantifying macroscopic coherence from an operational point of view, based on the relationship between the precision of the measurement and disturbance of the quantum state.


\end{abstract}
\submitto{\NJP}
\maketitle

\section{Introduction}
Schr\"odinger's cat paradox dramatically illustrates a macroscopic object being in a quantum superposition of two macroscopically different states \cite{Schrodinger35}.
Although this famous thought experiment depicts an extreme example, the existence of such
superpositions and entanglement at macroscopic levels is not excluded by quantum theory.
Considerable experimental efforts have gone on to push the envelope by superposing ever larger quantum systems \cite{Brune96,Monroe96,Arndt99,Ourjoumtsev07,  Schoelkopf14, Kasevich15}.
There have also been attempts to characterize and quantify quantumness in a macroscopic sense  \cite{Leggett80,Dur02,Shimizu02,Bjork04,Shimizu05,Cavalcanti06,Korsbakken07,Lee11, Frowis12,Sekatski14,Nimmrichter13,Yadin15, Frowis15, Jeong15,Park16,Yadin16}.
Several general measures for quantifying such \textit{quantum macroscopicity} have been suggested in recent studies  \cite{Lee11,Frowis12, Nimmrichter13}.
However, those measures tend to operate within quite different contexts such as 
distinguishability between component states with a finite measurement precision \cite{Korsbakken07,Sekatski14},
interference in the phase space \cite{Lee11}, usefulness for quantum metrology \cite{Frowis12}, and the minimal modification of quantum theory \cite{Nimmrichter13}.

Meanwhile, a resource theory of quantum coherence has  recently been proposed  \cite{Baumgratz14}.
In \cite{Baumgratz14}, the amount of quantum coherence could be quantified as a physical resource to achieve tasks beyond classical types of resources.
In this viewpoint, recent studies have discovered connections between quantum coherence and other fields of resource theory, including quantum correlation \cite{Ma15, Streltsov15, Xi15}, asymmetry \cite{Gour08,MarvianNC, Marvian14}, and quantum thermodynamics \cite{Cwiklinski15,LostaglioX,LostaglioNC}.
Recently, an axiomatic approach towards macroscopic quantum coherence was  suggested \cite{Yadin16} and several existing measures \cite{Lee11, Frowis12, Girolami14, Yadin15, Yadin16} were investigated based on it.

In this paper, we suggest a measure of macroscopic coherence based on the state disturbance induced by a coarse-grained measurement. 
We show that the disturbance-based measure satisfies recently proposed criteria of macroscopic coherence \cite{Yadin16}, but in some cases cannot yield consistent results without additional constraints.
This problem is overcome in our study by introducing coarse-graining of the measurement depending on the system size.
We prove an inequality which relates the Wigner-Yanase-Dyson skew information (and consequently, the quantum Fisher information) and the
state disturbance induced by coarse-grained measurement, from which we argue that an appropriate limit to yield a consistent measure is the classical
limit.
We further 
show that our concept of quantum macroscopicity corrsponds to the fragility of a quantum state under a certain type of decoherence.
 Our operational viewpoint on quantum macroscopicity allows one to effectively identify the quantum coherence between the macroscopically-separated components of a superposition.
Our approach can be applied to both spin and bosonic systems, and we present several examples that lead to reasonable results.

\section{Disturbance-Based Measure of Macroscopic Quantum coherence}
\subsection{Criteria of macroscopic quantum coherence}
We first review some preliminary concepts regarding macroscopic quantum coherence.
Let us consider a measurement observable  described by a hermitian operator $\hat{A} = \sum_i a_i \ket{i} \bra{i}$. The eigenstates of the observable $\hat{A}$ define a natural orthonormal basis $\{ \ket{i} \}$, which can be used to quantify the amount of coherence in the system.
Previous measures of quantum coherence \cite{Baumgratz14, Streltsov15} quantify the degree of coherence contained in the quantum state with respect to the given basis $\{ \ket{i} \}$.
However, these measures give the same value for every superposition in the form of $\ket{i} + \ket{j}$, without any regard for physical measurement outcomes represented by components $\ket{i}$ and $\ket{j}$, which are $a_i$ and $a_j$ respectively.
In other words, they did not consider how correctly $\ket{i}$ and $\ket{j}$ are discriminated by an actual measurement.
In an attempt to quantify macroscopic quantum coherence however, we should give some consideration to the 
outcomes of a physical measurement.

Recently, Yadin and Vedral  proposed \cite{Yadin16}  a set of conditions that should be satisfied by a proper measure of  macroscopic coherence.
In their proposed resource theory of macroscopic coherence, the free operation $\cal E$ is characterized as completely positive trace-nonincreasing operations satisfying the condition ${\cal E} ({\hat{\rho}}^{(\delta)}) = {\cal E}({\hat{\rho}})^{(\delta)}$, where ${\hat{\rho}}^{(\delta)} = \sum_{a_i - a_j = \delta} {\rho}_{ij} \ket{i}\bra{j}$.
Under such free operations, coherence terms $\ket{i} \bra{j}$ with different modes $\delta = a_i - a_j$ cannot be mixed together, by which a physical distance of superposition $|\delta|$ cannot be increased freely, i.e. a transition from $\ket{0} + \ket{1}$ to $\ket{0}+\ket{N}$ is prohibited when $a_1 \neq a_N$.
This type of free operations has been previously studied in the context of asymmetry in a quantum state \cite{Gour08, MarvianNC, Marvian14}.
With respect to this set of free operations, the authors of \cite{Yadin16} proposed that any reasonable measure of macroscopic quantum coherence $M({\hat{\rho}})$ based on the resource theory should satisfy the following conditions:\\

\begin{itemize}
\item [(M1)~] $M({\hat{\rho}}) \geq 0$ and $M({\hat{\rho}}) = 0$ {\it if and only if} ${\hat{\rho}} = {\hat{\rho}}^{(0)}$.

\item[(M2a)] Non-increasing under any trace-preserving free operation, $M({\cal E}({\hat{\rho}})) \leq M({\hat{\rho}})$.

\item[(M2b)] Non-increasing under any selective free operation,
$\sum_\alpha p_\alpha M({\cal E}_\alpha ({\hat{\rho}}) /p_\alpha) \leq M({\hat{\rho}})$ for ${\cal E} = \sum_\alpha  {\cal E}_\alpha$, where $p_\alpha = \Tr {\cal E}_a({\hat{\rho}})$.

\item[(M3)~] Convexity, $M(\sum_i p_i {\hat{\rho}}_i) \leq \sum_i p_i M({\hat{\rho}}_i)$.

\item [(M4)~] $M(\ket{i} + \ket{j}) > M(\ket{k} + \ket{l})$ if $|a_i - a_j | > |a_k - a_l|$.\\
\end{itemize}

Here, (M1) identifies free states which do not contain any macroscopic quantum coherence. 
(M2a) and (M2b) are required in a sense that one cannot increase macroscopic quantum coherence freely (i.e. by free operations),
and often called weak and strong monotonicity conditions, respectively.
The condition (M3)  guarantees that macroscopic quantum coherence does not increase by mixing quantum states.
Finally, condition (M4) is to quantify the macroscopic size of a superposition based on the distance  between component states in terms of the difference between corresponding eigenvalues.
This additional condition (M4) restricts the set of asymmetry monotones into a set of measures that discriminate macroscopic and microscopic superpositions. 
In this sense, the resource theory of \cite{Yadin16} may be understood as a type of  an asymmetry (M1-M3) in addition to a size factor (M4).
Yadin and Vedral  pointed out \cite{Yadin16}  that among  two general measures of quantum macroscopicity, one for bosonic systems \cite{Lee11} and the other for spin system \cite{Frowis12}, only the latter \cite{Frowis12} based on the quantum Fisher information satisfies all the conditions (M1)--(M4).
Known examples of measures that satisfy all these conditions are the quantum Fisher information and the Wigner-Yanase-Dyson skew information \cite{Yadin16}.

\subsection{Macroscopic coherence and coarse-grained measurement}

We say that macroscopic coherence is coherence of a quantum superposition between two macroscopically distinct states. In other words, the component states of the superposition are supposed to yield two distinct outcomes when a measurement on a macroscopic scale is performed.
We may employ the concept of a coarse-grained measurement \cite{Poulin05, Kofler07, Barnea16} to describe such a macroscopic measurement.
In order to construct a coarse-grained measurement, we first define a smoothing function $q_i^\sigma(x) = (\sqrt{2\pi}\sigma)^{-1} \exp[-(a_i-x)^2/(2\sigma^2)]$, where $x$ is a continuous variable over the real line. A natural choice for the smoothing function $q_i^\sigma(x)$ is a Gaussian distribution centered around measurement outcome $a_i$. The standard deviation $\sigma$ determines the level of precision of the measurement and therefore quantifies the amount of  coarse graining of the measurement. A \textit{coarse grained measurement} is then defined to be the following set of Kraus operators:
\begin{equation}
\hat{Q}^\sigma_x = \sum_{i} \sqrt{q_i^\sigma(x)} \ket{i}\bra{i}.
\end{equation}
One may interpret the above measurement as an interaction with the needle of a measuring apparatus that returns a normal distribution about the position $a_i$ when the system is in the space projected by $\hat{P}_i = \ket{i} \bra{i}$.
If $\sigma \rightarrow 0$, the measurement process becomes projective, while an increasing $\sigma$ implies an increasingly imprecise measurement process.
One may verify that $ \int^{\infty}_{-\infty} \hat{Q}_x^{\sigma \dag} \hat{Q}_x^\sigma dx = \mathbb 1$ for any $\sigma > 0 $ so it is indeed a valid positive-operator valued measurement (POVM).
In such a case, the post measurement state is given by $\Phi_\sigma({\hat{\rho}}) = \int_{-\infty}^\infty dx \hat{Q}_x^\sigma {\hat{\rho}} \hat{Q}_x^{\sigma \dagger} = \sum_{\delta \in \Delta} e^{-\delta^2/(8\sigma^2)} {\hat{\rho}}^{(\delta)}$, where $\Delta = \{a_i - a_j \}$ is a set of the spacing between the eigenvalues of the observable $\hat{A} = \sum_i a_i \ket{i} \bra{i}$.

\subsection{Quantum state disturbance under coarse-grained measurement}
There have been studies on quantifiers of  the size of a superposition based on the distinguishability between two components states with a finite measurement precision \cite{Korsbakken07, Korsbakken10, Sekatski14}.
Reference \cite{Sekatski14} suggested a measure of the size of macroscopic superpositions by quantifying the amount of noise that can be tolerated by a coarse-grained photon number measurement.
Applications of these measures, however, are limited  only to pure states and it is required to choose a specific decomposition (such as $|A\rangle+|B\rangle$) that represents the superposition.
Here, we show that the quantum state disturbance caused by a coarse-grained measurement naturally leads to measures of macroscopic coherence that are applicable to arbitrary forms of states and that satisfy all conditions (M1)--(M4). 

When one performs a non-selective projective (i.e. precise) measurement on the state with the given measurement basis set,
all coherence terms between eigenstates of the different measurement outcomes will vanish. However, when a coarse-grained measurement is performed, certain coherence terms may survive depending on the precision of the measurement.
 It is therefore reasonable to expect that at a certain level of the measurement precision, only macroscopic coherence will be disturbed by the measurement process.
Towards this end, we propose the disturbance of the quantum state induced by the coarse grained measurement process as a natural measure of macroscopic quantum coherence.

In order to quantify quantum macroscopicity by quantum state disturbance,
we will employ distance measures $D(\hat\rho, \hat\tau)$ between quantum states $\hat\rho$ and $\hat\tau$ that satisfy the following set of conditions.

\begin{itemize}
\item [(D1)~] $D(\hat\rho,\hat\tau) \geq 0 $, where the equality is saturated if and only if $\hat\rho = \hat\tau$

\item [(D2)~] Unitary invariance: $D(\hat{U}\hat\rho\hat{U}^\dagger, \hat{U}\hat\tau\hat{U}^\dagger) =D(\hat\rho, \hat\tau)$.

\item [(D3a)] Contractivity under a completely positive trace-preserving map ${\cal E}$, $D(\hat\rho,\hat\tau) \geq D({\cal E}(\hat\rho), {\cal E}(\hat\tau))$ (Note that ${\cal E}$ is not necessarily a free operation).

\item [(D3b)] $D(\hat\rho,\hat\tau) \geq \sum_{\alpha} p_\alpha D({\cal E}_{\alpha}(\hat\rho)/p_\alpha, {\cal E}_{\alpha}(\hat\tau)/p_\alpha)$, when $p_\alpha = \Tr{\cal E}_\alpha(\hat\rho) = \Tr{\cal E}_\alpha(\hat\sigma)$ and $\sum_\alpha {\cal E}_\alpha = {\cal E}$. (Note  that ${\cal E}_\alpha$ is not necessarily a free operation).

\item [(D4)~] Joint convexity : $D(\sum_i p_i \hat\rho_i, \sum_i p_i \hat\tau_i) \leq \sum_i p_i D(\hat\rho_i, \hat\tau_i)$.
\end{itemize}

Remarkably, despite starting from considerably different physical arguments, the following theorem shows that the measurement disturbance satisfies the set of conditions proposed by Yadin and Vedral \cite{Yadin16}.

\begin{theorem} [Disturbance-based measure of macroscopic quantum coherence]
\label{Ms}
For any coarse-grained measurement process $\Phi_\sigma$ with $\sigma>0$,
\begin{equation}
M_\sigma ({\hat{\rho}}) := D({\hat{\rho}}, \Phi_\sigma({\hat{\rho}}))
\end{equation}
satisfies (M1) -- (M4) when the distance measure $D(\hat\rho,\hat\tau)$ satisfies (D1) -- (D4).
\end{theorem}

Details and proofs can be found in the Appendix.
Theorem~\ref{Ms} allows us to define a new family of macroscopic quantum coherence measures parametrized by the measurement precision $\sigma$.
In the special case of $\sigma = 0$, this type of measure becomes a measure of coherence with respect to the eigenbasis $\{ \ket{i} \}$ of the observable, suggested in \cite{YadinX}, but does not satisfy (M4) anymore.
The Bures distance $D_B(\hat\rho, \hat\tau) = 2 - 2 \sqrt{ {\cal F}(\hat\rho,\hat\tau) }$ defined in terms of the fidelity between quantum states ${\cal F}({\hat{\rho}}, \hat{\tau}) = [\Tr\sqrt{\sqrt{\hat{\rho}} \hat{\tau} \sqrt{\hat{\rho}}}]^2$ and the quantum relative entropy defined by $S({\hat{\rho}}|| \hat{\tau}) = \Tr{\hat{\rho}}\ln{\hat{\rho}} - \Tr{\hat{\rho}}\ln \hat{\tau}$ are good examples satisfying all the conditions (D1) -- (D4).
For the rest of the paper, we focus on the the measure based on the Bures distance, $M^B_\sigma({\hat{\rho}}) = D_B({\hat{\rho}}, \Phi_\sigma({\hat{\rho}}))$. 

However, we observe that the disturbance-based measure $M_\sigma({\hat{\rho}})$ with certain values of $\sigma$  may lead to unreasonable conclusions even when it satisfies all the conditions in \cite{Yadin16}.
The following example shows that a product of microscopic superpositions  has a larger value of $M_\sigma$ than the Greenberger-–Horne-–Zeilinger(GHZ)-state when $\sigma$ is sufficiently small. This is contrary to our understanding and previous results \cite{Lee11, Frowis12, Park16} that the latter state is clearly in a macroscopic superposition while the former is not.

Consider a magnetization measurement on a system of $N$ spin-$1/2$ particles, of the same type studied by Poulin~\cite{Poulin05}.
The measurement is defined by a Hermitian operator $\hat{M} := \sum_{i=1}^N \hat{s}_z^{(i)}$ where $s_z^{(i)} := \hat{\sigma}_z / 2$ and $\hat{\sigma}_z$ is the standard Pauli Z operator.
The observable $\hat{M}$ represents a collective measurement of the overall spins rather than addressing each individual spin.
We compare the values of quantum macroscopicity measure $M(\hat\rho)$ between two different quantum states, a product state ${\hat{\rho}}_{\rm p}^N = \ket{\Psi^N_{\rm p}}\bra{\Psi^N_{\rm p}}$ with $\ket{\Psi^N_{\rm p}} = ( \cos\theta \ket{0} + \sin\theta\ket{1} ) ^ {\otimes N}$ and the GHZ-state  ${\hat{\rho}}_{\rm GHZ}^N = \ket{\Psi^N_{\rm GHZ}}\bra{\Psi^N_{\rm GHZ}}$ with $\ket{\Psi^N_{\rm GHZ}} = 2^{-1/2} \left( \ket{0}^{\otimes N} +  \ket{1}^{\otimes N} \right)$.
The state $\hat\rho_{\rm p}^N$ is a product of microscopic superpositions (between $\ket{0}$ and $\ket{1}$) and does not contain long range coherence between the spins in the system. Moreover, $\hat\rho_{\rm p}^N$ is a kind of a spin coherent state and its classicality has been studied in Refs.~\cite{Kofler07, Markham03}.
On the other hand, $\hat\rho_{\rm GHZ}^N$ 
could be a typical model of Schr\"odinger's cat state that
two components of the superposition give maximally different outcomes (all spin up, all spin down),
leading to a large variance for the observable $\hat{M}$.
Also it contains multipartite quantum correlation between the spins in the system \cite{Tichy16}.

In order to compare the quantum macroscpicity $M_\sigma^B(\hat\rho)$, we first evaluate the fidelity between the pre- and post-measurement states.
The fidelity for the product state ${\hat{\rho}}_{\rm p}^N$ is given by
${\cal F}({\hat{\rho}}_{\rm p}^N, \Phi_\sigma({\hat{\rho}}_{\rm p}^N)) \approx ({1 + N \sin^2(2\theta) /(8\sigma^2))^{-1/2}}$ 
using the approximation of the binomial distribution to the normal distribution for $N \gg 1$.
On the other hand, in the case of the GHZ-state ${\hat{\rho}}_{\rm GHZ}^N$, we have ${\cal F}({\hat{\rho}}_{\rm GHZ}^N, \Phi_\sigma({\hat{\rho}}_{\rm GHZ}^N)) =  2^{-1} (1 + \exp[{-N^2/(8\sigma^2)}])$.
Note that for small enough values of $\sigma \ll 1$, $M_\sigma({\hat{\rho}}_{\rm p}^{N})$  tends to the maximum value of $2$ for the product state, while $M_\sigma({\hat{\rho}}_{\rm GHZ}^N)$ is $2-\sqrt{2} \approx 0.586$ for the GHZ state (see \Fref{Spin}). This suggests that an accumulation (i.e., direct product) of microscopic superposition is more macroscopically-quantum than a pure superposition of two macroscopically distinct states.
The result clearly demonstrates that
the conditions proposed in \cite{Yadin16} are not sufficient to prescribe a completely consistent measure of macroscopic coherence.

\begin{figure}[t]
\centering
\includegraphics[width=0.6\linewidth]{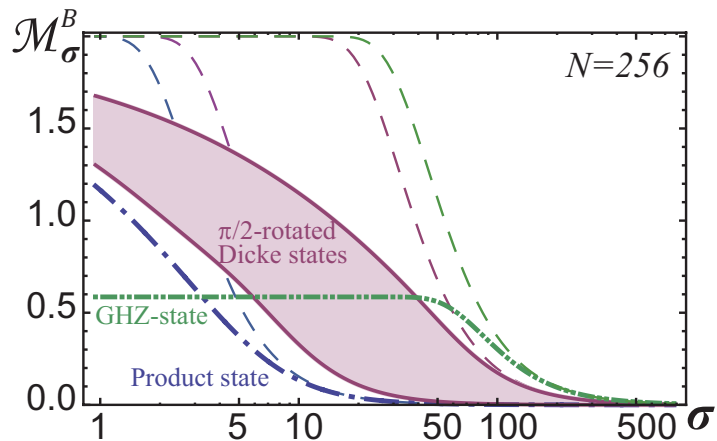}
\caption{Disturbance-based coherence $M_\sigma^B$ for measuring total magnetization of $N$ spin-1/2 system with $N=256$.
A product state (dot-dashed line), a GHZ-state (double-dot-dashed line), and $\pi/2$-rotated Dicke states (solid lines) are investigated.
Upper line on rotated Dicke states (shaded region) refers to $k=N/2$, while lower line refers to $k=1$. Dashed lines refer to the bound given by \eref{Ieq3}.}
\label{Spin}
\end{figure}

\section{Quantifying Macroscopic coherence}
\subsection{Quantum state disturbance and macroscopic coherence}
In order to overcome the issues described in the previous section, we revisit to the basic premise of macroscopic quantumness.
As far back as Schr\"odinger \cite{Schrodinger35}, a system is said to be macroscopic quantum
when each state constructing superposition is distinguished directly by a classical measurement.
In metrology, it is well known that the limit of a classical measurement is given by 
$\sigma \propto \sqrt{N}$ for $N$-particle systems, and quantum resources are necessary to achieve higher efficiencies 
\cite{Giovannetti04, Giovannetti06}.
 Previous studies of coarse-grained measurement similarly argued that the condition $\sigma \gg \sqrt{N}$ allows macroscopic observables to be considered classical \cite{Poulin05, Kofler07}.

The following theorem relating our disturbance-based measure $M^B_\sigma({\hat{\rho}})$ to the Wigner-Yanase-Dyson skew information $I_W({\hat{\rho}},\hat{A}) = (-1/2)\Tr[\sqrt{\hat{\rho}},\hat{A}]^2$ further reinforces our argument.

\begin{theorem} \label{skewInfo}
Coarse-grained measurement disturbance $M^B_\sigma({\hat{\rho}})$ is lower bounded by Wigner-Yanase-Dyson skew information $I_W({\hat{\rho}},\hat{A})$,
\begin{equation}
\label{Ieq2}
M^B_\sigma({\hat{\rho}}) \leq 2\left( 1 - e^{-\frac{I_W({\hat{\rho}}, \hat{A})}{4\sigma^2}}\right).
\end{equation}
For a pure state $\ket{\psi}$, we have
\begin{equation}
\label{Ieq3}
M^B_\sigma(\ket{\psi}) \leq 2 \left( 1 - e^{-\frac{{\rm Var}_{\ket{\psi}}(\hat{A})}{8\sigma^2}} \right),
\end{equation}
where ${\rm Var}_{\ket{\psi}} (\hat{A}) = \bra{\psi} \hat{A}^2 \ket{\psi} - \bra{\psi}\hat{A}\ket{\psi}^2$ is the variance of the observable $\hat{A}$, which is identical to $I_W(\ket{\psi}\bra{\psi}, \hat{A})$ for a pure state.
\end{theorem}

The above inequality reflects the intuition that 
the more precise the measurements and the more coherence present within the system,
the more the measurement will disturb the quantum state.

A previous study \cite{Frowis12} argued that scaling of the quantum Fisher information with the number of particles $N$ characterizes whether a $N$-particle system is macroscopically quantum.
Moreover, the Wigner-Yanase-Dyson skew information is a closely related with the quantum Fisher information due to the following relation \cite{Luo04}
\begin{equation}
\label{FWbound}
4 I_W({\hat{\rho}}, \hat{A}) \leq I_F({\hat{\rho}}, \hat{A}) \leq 8 I_W({\hat{\rho}}, \hat{A}),
\end{equation}
where the quantum Fisher information is given by $I_F({\hat{\rho}}, \hat{A}) = 2 \sum_{i \neq j} (\lambda_i - \lambda_j)^2/(\lambda_i + \lambda_j) |\langle \psi_i |\hat{A} | \psi_j \rangle|^2$ for eigendecomposition of ${\hat{\rho}} = \sum_i \lambda_i \ket{\psi_i} \bra{\psi_i}$.
We then note that the inequality (\ref{Ieq2}) relates our measure to the previous suggested measure of quantum macroscopicity based on the quantum Fisher information \cite{Frowis12}.
According to
\cite{Frowis12}, quantum states with $I_F({\hat{\rho}},\hat{A}) = O(N^1)$ can be interpreted as classical(or at least microscopic quantum) while the states with $I_F({\hat{\rho}},\hat{A}) = O(N^2)$ may be considered macroscopic quantum.

It is worth mentioning that a similar inequality was recently derived in a separate study of macrorealism based on the Leggett-Garg inequality \cite{Frowis16} as
\begin{equation}
\label{Ieq4}
\sqrt{{\cal F}(\hat\rho, \Phi_\sigma(\hat\rho))} \geq B_F = e^{- \frac{I_F (\hat\rho, \hat{A})} {32 \sigma^2}} - {\rm erfc}\left( \frac{\sqrt{2} \pi \sigma}{\sqrt{I_F (\hat\rho, \hat{A})}} \right)
\end{equation}
while the inequality \eref{Ieq2} can be expressed as $\sqrt{{\cal F}(\hat\rho, \Phi_\sigma(\hat\rho))} \geq B_W = e^{-\frac{I_W(\hat\rho,\hat{A})}{4\sigma^2}}$.
We point out that the bounds $B_F$ give negative values when $I_F/\sigma^2 \geq 37.806$, which leads to the trivial bound $B_F < 0 \leq \sqrt{{\cal F}(\hat\rho, \Phi_\sigma(\hat\rho))}$, while our bound $B_W$ is positive for any $I_W$ and $\sigma$.
In the case of a pure state $\ket{\psi}$, the bound given by \eref{Ieq3} is always tighter than the bound \eref{Ieq4} given by the quantum Fisher information, since $I_F(\ket{\psi}, \hat{A})  = 4 {\rm Var}_{\ket{\psi}} (\hat{A})$.
In the case of a mixed state, the bound $B_F$ seems tighter than $B_W$ when $\sigma$ is large. However, in some regions of small $\sigma$, $B_W$ could be tighter than $B_F$.
A detailed analysis with an example is presented in \ref{AppendixC}.

\subsection{Examples in spin and bosonic systems}
 Theorem~\ref{skewInfo} naturally manifests itself in the disturbance-based measure.
Provided the level of coarse graining is chosen to be $\sigma \propto \sqrt{N}$,
a state with $I_W({\hat{\rho}},\hat{A}) = O(N^1)$ will result in a measurement disturbance close to zero.
For example, the macroscopic coherence for a product of microscopic quantum states ${\hat{\rho}}^{\otimes N}$ is close to zero according to our measure, since the Wigner-Yanase-Dyson skew information scales with the order of $O(N)$.
In contrast, a non-classical skew information $I_W({\hat{\rho}},\hat{A}) = O(N^2)$, for example in the case of a GHZ state, allows the measure $M_\sigma^B({\hat{\rho}})$ to reach its maximum value of $2$ for $N \gg 1$.
This observation allows us to circumvent the inconsistency observed in the previous section.
We will therefore impose the classical 
limit $\sigma = \sqrt{N}$ as the appropriate level of coarse graining for our disturbance based measure.

Another example in the spin system is a rotated Dicke state given by $\hat{R}_{\theta,\phi} \ket{N,k}$, where $\ket{N,k} = {N \choose k}^{-1/2} \sum_P P( |\underbrace{0 \cdots 0}_{N-k} \underbrace{1\cdots 1}_{k}  \rangle)$ is a sum overall all symmetric permutations $P$, and $\hat{R}_{\theta, \phi} = e^{\xi \hat{J}_+ - \xi^* \hat{J}_- }$ is the rotation operator with $\hat{J}_\pm = \sum_{i=1}^{N} (\sigma_x^{(i)} \pm i \sigma_y^{(i)})$ and  $\xi = \theta e^{i\phi} /2 $. 
In the case of $\theta = \pi/2$ and $\phi = 0$, the macroscopic coherence of the state depends on  the excitation number $k$.
Such a state becomes a product state when $k = 0$ or $k = N$.

\Fref{Spin} compares the behavior of $M_\sigma^B$ between rotated Dicke, GHZ and product states for varying levels of the coarse graining parameter $\sigma$. 
We also observe that at the classical limit of $\sigma = \sqrt{N}$, rotated Dicke states with excitation number $k\approx N/2$ result in higher levels of macroscopic coherence $M^B_\sigma$ than the GHZ-state. 
This property does not persist however, if we were to continue decreasing the amount of measurement precision (i.e. increase $\sigma$). For a sufficiently large $\sigma$, the GHZ-state tends to have the highest level of macroscopic coherence among all the states considered.
Our disturbance-based measure appears to capture ideas from both the more general quantum coherence measures given by \cite{Baumgratz14} and the macroscopic coherence measures based on the variance of the observable \cite{Lee11, Frowis12, Girolami14, Yadin15} since it encodes information about {\it how many  states are currently in superposition} as well as {\it how far apart these superposed states are} with respect to the given measurement observable and the measurement precision $\sigma$.

\begin{figure}[t]
\centering
\includegraphics[width=0.6\linewidth]{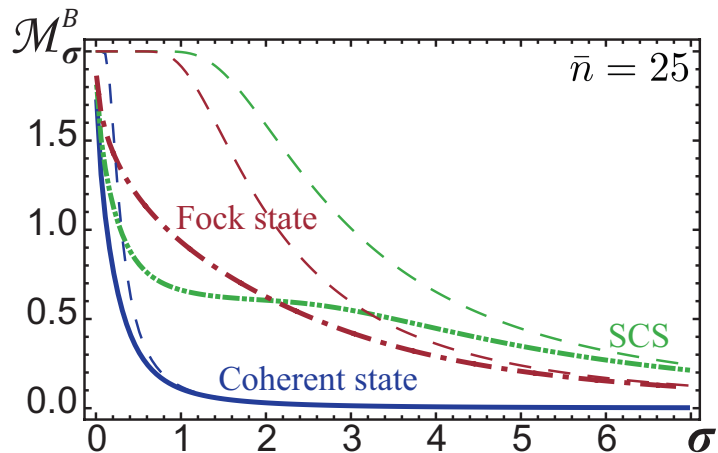}
\caption{
Disturbance-based coherence measure $M_\sigma^B$ for quadrature measurement for bosonic system with the same mean particle number $\bar{n} = 25$. 
A Fock state (dot-dashed line), a superposition of coherent states (double-dot-dashed line), and a coherent state (solid line) are investigated.
Dashed lines refer to the bound given by \eref{Ieq3}.}
\label{Boson}
\end{figure}

We also apply the disturbance-based measure to bosonic systems described by the annihilation operator $\hat{a}$ and the creation operator $\hat{a}^\dagger$.
Since a bosonic system can contain many particles in a single mode, the system may be considered macroscopic when the mean particle number $\bar{n} = \langle \hat{a}^\dagger \hat{a} \rangle$ is large.
In this case, the particle number $\hat{n} = \hat{a}^\dagger \hat{a}$ and the quadrature  $\hat{X}_\theta = (e^{-i\theta} \hat{a} + e^{i\theta} \hat{a}^\dagger )/\sqrt{2}$ are natural candidates for measurement observables.
We now consider the value of $M_\sigma^B$ with respect to an $X$-quadrature measurement $\hat{X} = (\hat{a} + \hat{a}^\dagger)/\sqrt{2}$.
\Fref{Boson} shows the disturbance-based measure $M_\sigma^B$ for typical states of a bosonic system.
Again, we see that for small values of $\sigma$, a bosonic coherent state $\ket{\alpha}$ contains non-trivial macroscopic quantumness, which are not in agreement with our understanding. However, $M_\sigma^B$ rapidly decreases with $\sigma$ and becomes essentially zero at the imposed classical limit of $\sigma \approx 1 $.
This makes sense when we note that bosonic coherent states are the most classical states among all pure states \cite{Mandel86, Zurek93}
and the classical measurement is based on electric (or magnetic) fields which are proportional to $\hat{X}_\theta$.
In the case of the $\hat{X}$ measurement for coherent states,
the noise is given by ${\rm Var} (\hat{X}) = 1/2$, while the signal is given by $\langle \hat{X} \rangle \sim \sqrt{\bar{n}}$.
Then the signal to noise ratio $\langle \hat{X} \rangle / \sqrt{{\rm Var} (\hat{X}) }$ scales by $\sqrt{\bar{n}}$ which corresponds to the measuring of the magnetization $\hat{M}$ for the spin system with spin coherent states, $\langle \hat{M} \rangle / \sqrt{{\rm Var} (\hat{M})} \propto \sqrt{N}$.
Based on this, the noise corresponding to the bosonic system, $\sigma \sim \sqrt{{\rm Var} (\hat{X}) } \sim 1$, would be a proper choice of the classical limit.

We also evaluate the values of $M_\sigma^B$ for a superposition of coherent states (SCS) $\ket{\alpha} + \ket{-\alpha}$ and the Fock state $\ket{n}$.
In the phase space, the distance between  two bosonic coherent states, $\ket{\alpha}$ and $\ket{-\alpha}$, becomes greater when amplitude $\alpha$ becomes larger. The two coherent states can then be distinguishable by a ``classical-like'' measurement such as  a homodyne detection with a large degree of imprecision.
Thus, a SCS for $\alpha\gg1$ is often exemplified as a typical example of a macroscopic superposition and even called a Schr\"odinger cat state.
It may not be immediately clear whether Fock states $\ket{n}$ are macroscopic superpositions. However, in the coherent state representation, a Fock state of $n\gg1$ can also be understood as a superposition of many coherent states where the coherent states are far separate in the phase space. So, they may be possible candidates for macroscopic superpositions when $n\gg1$.

 In comparison to coherent states, a SCS and the Fock state give non-trivial values of $M_\sigma^B$ at the classical limit of $\sigma$ (see \Fref{Boson}).
All these observations are compatible with the common expectation that coherent states are classical, while SCS and the Fock states are considered macroscopically quantum.

\subsection{Connection to a decoherence model}
Decoherence in a particular basis can be regarded as a measurement performed by the environment \cite{Sekatski14PRL}.
Based on this concept, we may consider a connection between the quantum macroscopicity measure in the present work and the fragility of a quantum state by a certain type of decoherence.
We show that a coarse-grained measurement of observable $\hat{A}$
can be equivalently modeled by a decoherence process under linear coupling between the system observable $\hat{A}$ and the environment operator $\hat{p}_E$.
After time $t$, the initial state of the system $\hat\rho_0$ evolves into
\begin{equation}
\hat\rho(t) = \Tr_E  e^{-i g \hat{A} \hat{p}_E t}  \left( \hat\rho_0 \otimes \hat\tau_E \right) e^{i g \hat{A} \hat{p}_E t},
\end{equation}
where $\hat\tau_E$ is the initial state of the environment and $g$ is a coupling constant.
By taking the eigenstates of $\hat{p}_E$ to be $\ket{p}$ with continuous variable $p$, we get
\begin{equation}
\hat\rho(t) = \sum_{\delta \in \Delta} \hat\rho^{(\delta)}_0 \left[ \int dp  e^{-i (g t)   p \delta  } \bra{p} \hat\tau_E \ket{p} \right].
\end{equation}
Now we choose the environment state to have
$\bra{p} \hat\tau_E \ket{p} \propto e^{-\mu^2 p^2}$ so that
$\hat\rho(t) = \sum_{\delta \in \Delta} \hat\rho^{(\delta)}_0 e^{-\frac{(gt)^2 \delta^2}{4\mu^2}}$.
In this case, the state distance between the initial and final states $D(\hat\rho_0, \hat\rho(t))$, which indicates the fragility of the initial quantum state under this kind of decoherence, is exactly the state disturbance $M_\sigma(\hat\rho_0)$ caused by a coarse-grained measurement for the corresponding value of $\sigma = \mu / (\sqrt{2} gt)$.

For example, we suppose that the environment is in a thermal state $\hat\tau_E = e^{-\beta \hat{H}_E} / Z_E$, where $Z_E = \Tr e^{-\beta \hat{H}_E}$ and $\beta = (k_B T)^{-1}$ is an inverse temperature.
For simplicity, we further assume that the thermal bath is a single-mode harmonic oscillator with hamiltonian $\hat{H}_E = \hbar \omega \left( \hat{a}^\dagger_E \hat{a}_E + \frac{1}{2} \right)$ and the coupling with the system is given by the momentum operator, $\hat{p}_E= (\hat{a}_E - \hat{a}^\dagger_E)/(\sqrt{2} i)$.
In this case, we have
$\bra{p} \hat\tau_E \ket{q}  \propto \exp[- \tanh{(\beta\omega/2)} \hat{p}_E^2]$,
and $\tanh(\beta\omega/2)/(g^2 t^2)$ corresponds to $2 \sigma^2$ in the coarse-grained measurement.
We then see that large values of $\sigma$ correspond to short decoherence times, weak coupling and/or low bath temperatures.
In other words, a quantum state with a large value of $M_\sigma(\hat\rho)$ for the classical limit of $\sigma$ 
is easily decohered by a thermal environment.

\section{Conclusion}
We proposed a disturbance-based measure of macroscopic coherence through coarse grained-measurements.
Our argument stems from physical grounds that a precise measurement will affect all the coherence present in the system, while a sufficiently imprecise measurement will affect only the portion of the coherence between classically distinct states.
We demonstrated that our disturbance-based measure satisfies 
a series of properties to quantify macroscopic coherence 
laid out in \cite{Yadin16}.
In the process, we pointed out that 
conditions for macroscopic coherence proposed in \cite{Yadin16}
is insufficient to yield consistent results
without additional constraints.
This inconsistency can be overcome by fixing the level of coarse-graining to an appropriate classical limit.
We also demonstrated an inequality relating the measurement-induced disturbance and the Wigner-Yanase-Dyson skew information 
and argued that this kind of classical limit is necessary to produce a reliable measure of macroscopic coherence.
Furthermore, we establish the direct connection between the disturbance-based quantum macorscopicity measure and the fragility of a quantum state under decoherence.

We emphasize that the proposed measure provides an operational point of view on macroscopic quantumness that can be quantified by the degree of disturbance throughout the measurement of a given imprecision.
The imprecision of the measurement allows us to focus on the coherence between macroscopically distinct states by blurring the interference below the measurement resolution.
We can thus identify whether the quantum state is in a macroscopic superposition by investigating the state disturbance throughout the measurement only with a macroscopic resolution.
As we have demonstrated for both spin and bosonic systems, our approach is not limited to a specific quantum system but can be applied to arbitrary macroscopic observables and quantum systems with large particle numbers.
We expect that the viewpoint concerning the state disturbance induced by coarse-grained measurement may lead to greater insights on macroscopic quantum effects and coherence.

\section*{Acknowledgment}
This work was supported by the National Research Foundation of Korea (NRF) through a grant funded by the Korea government (MSIP) (Grant No.~2010-0018295) and  by the KIST Institutional Program (Project No.~2E26680-16-P025). H. K. was supported by the Global Ph.D. Fellowship Program through the NRF funded by the Ministry of Education (Grant No. 2012-003435).
 K.C. Tan was supported by Korea Research Fellowship Program through the National Research Foundation of Korea (NRF) funded by the Ministry of Science, ICT and Future Planning (Grant No. 2016H1D3A1938100).

\appendix
\section{Proof of Theorem.~\ref{Ms}}
In this section, we prove that for a distance measure satisfying (D1) -- (D4), $M_\sigma(\hat\rho)  = D (\hat\rho, \Phi_\sigma(\hat\rho))$
satisfies the conditions (M1) -- (M4).
We first prove the following proposition:
\begin{proposition} 
\label{EP}
$\delta$-coherence preserving operation ${\cal E}$ commutes with any coarse-grained measurement process $\Phi_\sigma$ for any state ${\hat{\rho}}$, i.e.
\begin{equation}
({\cal E} \circ \Phi_\sigma) ({\hat{\rho}}) = (\Phi_\sigma \circ {\cal E}) ({\hat{\rho}}).
\end{equation}
\end{proposition}
\begin{proof}
By using the property of the free operation ${\cal E}$,
\begin{equation}
\eqalign{
({\cal E} \circ \Phi_\sigma) ({\hat{\rho}}) &= \sum_{\delta \in \Delta} e^{-\frac{\delta^2}{8\sigma^2}} {\cal E} ({\hat{\rho}}^{(\delta)}) \\
 &= \sum_{\delta \in \Delta} e^{-\frac{\delta^2}{8\sigma^2}} {\cal E} ({\hat{\rho}})^{(\delta)} \\
&= (\Phi_\sigma \circ {\cal E}) ({\hat{\rho}}).}
\end{equation}
\end{proof}

We now prove for conditions (M1) -- (M4).
\begin{proof}
(M1) Note that $M_\sigma({\hat{\rho}}) = D(\hat\rho, \Phi_\sigma(\hat\rho)) = 0 $ if and only if ${\hat{\rho}} = \Phi_\sigma ({\hat{\rho}})$. This is only achieved when $\Phi_\sigma$ is given by a convex sum of projections $\hat{P}_n = \sum_{a_i = n} \ket{i}\bra{i}$, thus this condition can be achieved when ${\hat{\rho}} = {\hat{\rho}}^{(0)}$.

(M2a) By using Proposition.~\ref{EP}, we show that
\begin{equation}
\eqalign{
M_\sigma ({\cal E} ({\hat{\rho}}) ) &= D( {\cal E} ({\hat{\rho}}), (\Phi_\sigma \circ {\cal E}) ({\hat{\rho}})) \\
&= D( {\cal E} ({\hat{\rho}}), ({\cal E} \circ \Phi_\sigma ) ({\hat{\rho}})) \\
&\leq D( {\hat{\rho}}, \Phi_\sigma ({\hat{\rho}})) \\
&=M_\sigma({\hat{\rho}})}
\end{equation}
for trace-preserving free operation ${\cal E}$.

(M2b) Similarly, by using Proposition.~\ref{EP} and the condition (D3b), we show that
\begin{equation}
\eqalign{
\sum_\alpha p_\alpha M_\sigma ({\cal E}_\alpha ({\hat{\rho}})/p_\alpha)
&= \sum_\alpha p_\alpha D({\cal E}_\alpha ({\hat{\rho}})/p_\alpha,  \Phi_\sigma ({\cal E}_\alpha ({\hat{\rho}})/p_\alpha))\\
&= \sum_\alpha p_\alpha D({\cal E}_\alpha ({\hat{\rho}})/p_\alpha,  {\cal E}_\alpha ( \Phi_\sigma({\hat{\rho}}))/p_\alpha)\\
&\leq D ({\hat{\rho}}, \Phi_\sigma({\hat{\rho}}))\\
&=M_\sigma({\hat{\rho}}),}
\end{equation}
where $\Tr{\cal E}_\alpha(\Phi_\sigma(\hat\rho)) = \Tr \Phi_\sigma ({\cal E}_\alpha(\hat\rho)) = \Tr {\cal E}_\alpha(\hat\rho) =p_\alpha$, since $\Phi_\sigma$ is a trace preserving map.

(M3) Convexity can be directly proven by using joint convexity of the distance measure,
$M_\sigma(\sum_i p_i \hat\rho_i) = D(\sum_i p_i \hat\rho_i,\Phi_\sigma(\sum_i p_i \hat\rho_i)) = D(\sum_i p_i \hat\rho_i, \sum_i p_i\Phi_\sigma(\hat\rho_i)) \leq \sum_i p_i D(\hat\rho_i, \Phi_\sigma(\hat\rho_i)) = \sum_i p_i M_\sigma(\hat\rho_i)$.

(M4) 
Now we prove that there is an ordering of $M$ between two states
$\ket{\psi_0} = (\ket{i} + \ket{j})/\sqrt{2}$ and $\ket{\psi_1} = (\ket{k} + \ket{l})/\sqrt{2}$, i.e. $M(\ket{\psi_0}) > M(\ket{\psi_1})$ in the case of $|a_i - a_j| > |a_k - a_l|$.
Note that we can always choose the unitary operation $\hat{U}$, which transforms the bases $\ket{k} \rightarrow \ket{i}$ and $\ket{l} \rightarrow \ket{j}$.
Then in $\{\ket{i}, \ket{j}\}$ basis, we can write the states,
$
\hat\rho_0 = \ket{\psi_0}\bra{\psi_0} = \frac{1}{2} \left( 
\eqalign{
1 ~~ 1 \\
1 ~~ 1 
}
\right)
$, $\Phi_\sigma (\hat\rho_{0}) = \frac{1}{2} \left( 
\eqalign{
1 ~~ K \\
K ~~ 1 
}
\right)
$, 
$
\hat{U} \hat\rho_1 \hat{U}^\dagger = \hat\rho_0
$, and 
$
\hat{U} \Phi_\sigma (\hat\rho_1) \hat{U}^\dagger = \frac{1}{2} \left( 
\begin{array}{cc}
1 & J \\
J & 1 
\end{array}
\right)
$, where $K = e^{-(a_i-a_j)^2/(8\sigma^2)}$ and $J = e^{-(a_k-a_l)^2/(8\sigma^2)}$, respectively.
Since $|a_i - a_j| > |a_k - a_l|$,  $0 \leq K < J \leq 1$ for any $\sigma>0$, then we can choose $0 < \lambda = \frac{1-J}{1-K} <1$ such that
$ \hat{U} \Phi_\sigma (\hat\rho_1) \hat{U}^\dagger = \lambda \Phi_\sigma(\hat{\rho}_0)  + (1-\lambda) \hat\rho_0 $.
Then by the unitary invariance and joint convexity of the distance measure $D$, we have 
\begin{equation}
\eqalign{
M(\hat\rho_1)  
& =D(\hat\rho_1, \Phi_\sigma(\hat\rho_1)) \\
& = D(\hat{U} \hat\rho_1 \hat{U}^\dagger, \hat{U} \Phi_\sigma(\hat\rho_1) \hat{U}^\dagger) \\
& = D(\lambda \hat\rho_0 + (1- \lambda) \hat\rho_0, \lambda \Phi_\sigma(\hat{\rho}_0)  + (1-\lambda) \hat\rho_0) \\
&\leq \lambda D(\hat\rho_0, \Phi_\sigma(\hat\rho_0)) + (1-\lambda) D(\hat\rho_0,\hat\rho_0) \\
&= \lambda D(\hat\rho_0, \Phi_\sigma(\hat\rho_0))\\
& < D(\hat\rho_0, \Phi_\sigma(\hat\rho_0)) \\
& = M(\hat\rho_0),}
\end{equation}
which completes the proof.
\end{proof}

We also note that if a distance measure $D(\hat\rho,\hat\sigma)$ satisfies all conditions (D1)--(D4) without the property (D3b), the macroscopicity measure $M_\sigma(\hat\rho) = D(\hat\rho,\Phi_\sigma(\hat\rho))$ based on $D$ satisfies (M1)--(M4) except (M2b).

In the case of the Bures distance, $D_B(\hat\rho, \hat\sigma) = 2 - 2 \sqrt{ {\cal F}(\hat\rho,\hat\sigma) }$, (D1), (D2), (D3a), and (D4) can be easily proven by using the properties of the fidelity \cite{NielsenChuang}.
We can also prove the condition (D3b) by follows. 
\begin{equation}
\eqalign{
\sum_\alpha p_\alpha D_B({\cal E}_\alpha (\hat\rho)/p_\alpha, {\cal E}_\alpha(\hat\sigma)/p_\alpha)
&= 2 \sum_\alpha p_\alpha \left(1- \sqrt{{\cal F}({\cal E}_\alpha (\hat\rho)/p_\alpha, {\cal E}_\alpha(\hat\sigma)/p_\alpha)}\right)\\
&= 2 \left(1 - \sum_\alpha \sqrt{{\cal F}({\cal E}_\alpha (\hat\rho), {\cal E}_\alpha(\hat\sigma))} \right),
}
\end{equation}
by using $\sqrt{ {\cal F}(\hat\rho/p, \hat\sigma/p) } = \sqrt{ {\cal F}(\hat\rho, \hat\sigma) } / p$.
In order to complete the proof, we prove that 
$\sqrt{{\cal F}({\hat{\rho}}, \hat\sigma)} \leq \sum_\alpha \sqrt{{\cal F}({\cal E}_\alpha ({\hat{\rho}}), {\cal E}_\alpha (\hat\sigma))}$ when $\sum_\alpha {\cal E}_\alpha = {\cal E}$.
Note that ${\cal E}_\alpha({\hat{\rho}})$ can be expressed using ancillary state $\hat{\tau}_2$: ${\cal E}_\alpha ({\hat{\rho}}) = \Tr_2 ({\mathbb 1} \otimes \hat{\Pi}_\alpha) \hat{U} ({\hat{\rho}} \otimes \hat{\tau}_2) \hat{U}^\dagger  ({\mathbb 1} \otimes \hat{\Pi}_\alpha)$.
Note that fidelity is non-decreasing under partial trace $\sqrt{{\cal F}({\hat{\rho}}_{12}, \hat{\sigma}_{12})} \leq \sqrt{{\cal F}(\Tr_2 {\hat{\rho}}_{12}, \Tr_2 \hat{\sigma}_{12})}$ and satisfies following properties for a set of projection operators $\{ \hat{\Pi}_n \}$: 
$\sum_n \sqrt{{\cal F} (\hat{\Pi}_n {\hat{\rho}} \hat{\Pi}_n, \hat{\Pi}_n \hat{\sigma} \hat{\Pi}_n)} = \sqrt{{\cal F} (\sum_n \hat{\Pi}_n {\hat{\rho}} \hat{\Pi}_n, \sum_n \hat{\Pi}_n \hat{\sigma} \hat{\Pi}_n)}$.
Using these properties we can show that
$\sum_\alpha \sqrt{{\cal F}({\cal E}_\alpha ({\hat{\rho}}), {\cal E}_\alpha (\hat{\sigma}))} \geq \sqrt{{\cal F}(\hat{U} ({\hat{\rho}} \otimes \hat{\tau}_2) \hat{U}^\dagger,\hat{U} (\hat{\sigma} \otimes \hat{\tau}_2) \hat{U}^\dagger)} = \sqrt{{\cal F}({\hat{\rho}}, \hat{\sigma})}$, since fidelity is invariant under unitary operations.

The conditions (D1)--(D4) for the relative entropy $S({\hat{\rho}}||{\hat{\sigma}})$ can be proved similarly. (D1), (D2), (D3a), and (D4) directly comes from the elementary properties of relative entropy \cite{NielsenChuang}.
(D3b) can be proved by noting that 
$\sum_\alpha p_\alpha S({\cal E}_\alpha (\hat\rho)/p_\alpha || {\cal E}_\alpha(\hat\sigma)/p_\alpha) = \sum_\alpha S({\cal E}_\alpha (\hat\rho) || {\cal E}_\alpha(\hat\sigma) )$.
Then (D3b) can be proved a same argument above by using the property $\sum_n S(\hat{\Pi}_n {\hat{\rho}} \hat{\Pi}_n||\hat{\Pi}_n \hat\sigma \hat{\Pi}_n) = S(\sum_n \hat{\Pi}_n {\hat{\rho}} \hat{\Pi}_n|| \sum_n  \hat{\Pi}_n \hat\sigma \hat{\Pi}_n)$.

\section{Proof of Theorem.~\ref{skewInfo}}
\begin{proof}
In order to prove the upper bound of  $M^B_\sigma({\hat{\rho}}) = 2 - 2 \sqrt{{\cal F}(\hat\rho, \Phi_\sigma(\hat\rho))}$, we show the lower bound of the fidelity $\sqrt{{\cal F}(\hat\rho, \Phi_\sigma(\hat\rho))}$.
We first prove the inequality when the state is pure.
Note that when one of the states are pure, the fidelity is given by ${\cal F}(\ket{\psi} \bra{\psi}, \hat\sigma) = \bra{\psi} \hat\sigma\ket{\psi}$.
Then, we have for coarse-grained measurement process,
\begin{equation}
\eqalign{
{\cal F}(\ket{\psi}\bra{\psi}, \Phi_\sigma \left( \ket{\psi} \bra{\psi} \right) ) \\
\qquad = \int_{-\infty}^\infty dx \bra{\psi} \hat{Q}_x^\sigma \ket{\psi} \bra{\psi} \hat{Q}_x^\sigma \ket{\psi}\\
\qquad =\sum_{i,j} \exp{\left[-\frac{(a_i-a_j)^2}{8\sigma^2}\right]} \bra{\psi} {i} \rangle \langle {i} \ket{\psi} \bra{\psi} {j} \rangle \langle {j} \ket{\psi}\\
\qquad  \geq \exp{\left[ -\frac{\sum_{i,j} (a_i-a_j)^2 |\bra{\psi} {i} \rangle|^2 |\bra{\psi} {j} \rangle|^2}{8\sigma^2} \right]}.
}
\end{equation}
Note that $|\bra{\psi} i \rangle|^2$ is the probability of getting outcome $a_i$, thus $\sum_{i,j} (a_i-a_j)^2 |\bra{\psi} {i} \rangle|^2 |\bra{\psi} {j} \rangle|^2 = 2 {\rm Var}_{\ket{\psi}}(\hat{A})$, then  $\sqrt{{\cal F}(\hat\rho, \Phi_\sigma(\hat\rho))} \geq \exp[- {\rm Var}_{\ket{\psi}}(\hat{A})/(8\sigma^2)]$.

For mixed states, we use the inequality $\sqrt{ {\cal F} (\hat{\rho}, {\hat{\rho}}')} \geq \Tr \sqrt{\hat{\rho}} \sqrt{{\hat{\rho}}'}$, which can be proved by using Araki-Lieb-Thirring inequality \cite{Araki90}.
We also note that for an unital operator set $\{ \hat{K}_n \}$, $\sqrt{\sum_n \hat{K}_n {\hat{\rho}} \hat{K}_n^\dagger} \geq \sum_n \hat{K}_n \sqrt{\hat{\rho}} \hat{K}_n^\dagger$ by operator Jenesn's inequality \cite{Hansen03}.
Putting these inequalities together, we finally get $\sqrt{{\cal F}({\hat{\rho}},\sum_n \hat{K}_n {\hat{\rho}} \hat{K}_n^\dagger)} \geq \sum_n \Tr \sqrt{\hat{\rho}} \hat{K}_n \sqrt{\hat{\rho}} \hat{K}_n^\dagger$.
We also note that a coarse-grained measurement operator $\{ \hat{Q}^\sigma_x \} = \{ \sum_i \sqrt{q^\sigma_i(x)} \ket{i} \bra{i} \}$  is a unital operator set.
Then using previous results on the pure state, we get
\begin{equation}
\eqalign{
\sqrt{{\cal F}({\hat{\rho}}, \Phi_\sigma({\hat{\rho}}) )} 
&\geq \int_{-\infty}^{\infty} dx \Tr \left[ \sqrt{\hat{\rho}} \hat{Q}_x^\sigma \sqrt{{\hat{\rho}}} \hat{Q}_x^\sigma \right]\\
&= \sum_{i,j} \exp{\left[-\frac{(a_i - a_j)^2}{8\sigma^2}\right]} |(\sqrt{{\rho}})_{ij}|^2.}
\end{equation}
Now we apply Jensen's inequality in order to obtain the final result,
\begin{equation}
\eqalign{
\sqrt{{\cal F}({\hat{\rho}}, \Phi_\sigma({\hat{\rho}}) )} 
&\geq \exp{\left[-\frac{\sum_{i,j} (a_i - a_j)^2  |(\sqrt{\rho})_{ij}|^2}{8\sigma^2}\right]} \\&= \exp{\left[-\frac{I_W({\hat{\rho}},\hat{A})}{4\sigma^2}\right]},}
\end{equation}
where $I_W({\hat{\rho}}, \hat{A}) = (-1/2) \Tr[\sqrt{\hat{\rho}}, \hat{A}]^2 = (1/2)\sum_{i,j} (a_i - a_j)^2  |(\sqrt{\rho})_{ij}|^2$ is the Wigner-Yanase-Dyson skew information.
\end{proof}

\section{Comparison between the bound $B_F$ and $B_W$}
\label{AppendixC}

\begin{figure}[t]
\centering
\includegraphics[width=0.6\linewidth]{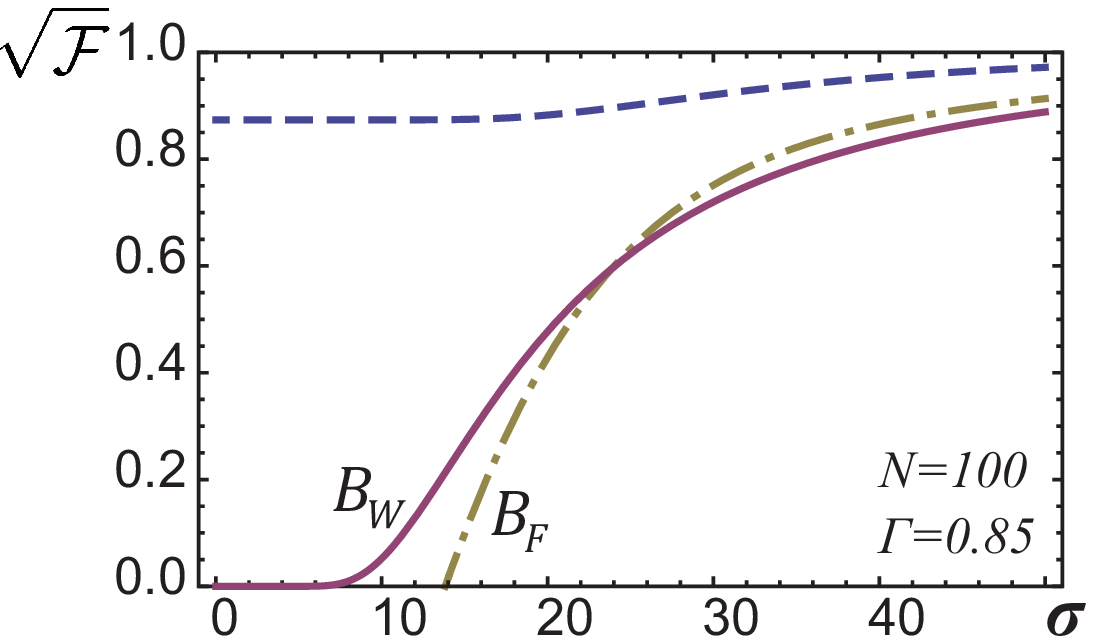}
\caption{Comparison between the lower bounds $B_F$(dot-dashed line) and $B_W$(solid line) of the fidelity $\sqrt{\cal F}$(dashed line) for the decohered GHZ-state $\hat\rho_\Gamma^N$ with $N=100$ and $\Gamma = 0.85$.
}
\label{CompareBound}
\end{figure}
We compare two different lower bounds of the fidelity between pre- and post measurement states, $\sqrt{{\cal F} (\hat\rho, \Phi_\sigma(\hat\rho))} \geq B_{F(W)}$ where $B_F=e^{- \frac{I_F (\hat\rho, \hat{A})} {32 \sigma^2}} - {\rm erfc}\left( \frac{\sqrt{2} \pi \sigma}{\sqrt{I_F (\hat\rho, \hat{A})}} \right)$  \cite{Frowis16} and $B_W=e^{-\frac{I_W(\hat\rho,\hat{A})}{4\sigma^2}}$, respectively.
We evaluate both the bounds
for a decohered GHZ-state in a spin system given by
\begin{equation}
\hat\rho_\Gamma^N = \frac{1}{2} \left( \ket{0} \bra{0} + \ket{N}\bra{N} +  \Gamma ( \ket{0}\bra{N}+  \ket{N}\bra{0}) \right),
\end{equation}
where $0 \leq \Gamma \leq 1$.
In this case, the quantum Fisher information and Wigner-Yanase-Dyson skew information are given by 
$ I_F(\hat\rho_\Gamma^N, \hat{A}) = N^2 \Gamma^2$ and $I_W(\hat\rho_\Gamma^N, \hat{A}) = \frac{N^2}{4} (1-\sqrt{1-\Gamma^2})$, respectively.
When the coarse-grain parameter $\sigma$ is large, the second term of $B_F$ becomes negligible and $B_F$ gives a tighter bound than $B_W$ (note that $I_F \leq I_W$ for any $\Gamma$).
On the other hand, if the coarse-grain parameter $\sigma$ is relatively small compared to $I_F$ and $I_W$, the second term of $B_F$ has a significant value while the first term becomes small.
Thus, there is some value of $\sigma$ where the two bounds, $B_F$ and $B_W$, meet as described in \Fref{CompareBound}.
When the state is pure ($\Gamma = 1$), we can use the bound \eref{Ieq3}, and our bound given by  $\sqrt{{\cal F} (\hat\rho, \Phi_\sigma(\hat\rho))} \geq B = e^{-\frac{{\rm Var}_{\ket{\psi}}(\hat{A})}{8\sigma^2}} = e^{-\frac{N^2}{32\sigma^2}}$ is always tighter than $B_F = e^{-\frac{N^2}{32\sigma^2}} - {\rm erfc}\left( \frac{\sqrt{2} \pi \sigma}{N^2} \right)$.

\end{document}